\documentclass{llncs}
\usepackage{amssymb}
\usepackage{amsfonts}
\usepackage{amsmath}
\usepackage[numbers]{natbib}
\usepackage{color}
\usepackage{tikz}
\usepackage{url}
\usepackage{bbm}

\newcommand{\R}{\mathcal{R}}

\newcommand{\AutoAdjust}[3]{\mathchoice{ \left #1 #2  \right #3}{#1 #2 #3}{#1 #2 #3}{#1 #2 #3} }
\newcommand{\Xcomment}[1]{{}}

\newcommand{\InBrackets}[1]{\AutoAdjust{[}{#1}{]}}%
\newcommand{\Ex}[2][]{\operatorname{\mathbf E}_{#1}\InBrackets{#2}}

\newcommand{\Indic}[1]{\mathbf{1}_{#1}}

\newcommand{\spot}{s}
\newcommand{\reservation}{r}

\newcommand{\action}{a}

\newcommand{\actionreservestage}{\action^{\reservation}}

\newcommand{\val}{v}
\newcommand{\vali}{\val_i}
\newcommand{\valagent}{\vali}

\newcommand{\typeCDF}{F}
\newcommand{\typePDF}{f}

\newcommand{\valdist}{F}
\newcommand{\valCDF}{\valdist}

\newcommand{\spotCDF}{S}

\newcommand{\strat}{s}
\newcommand{\strats}{\mathbf{\strat}}

\newcommand{\budget}{b}

\newcommand{\budgeti}{\budget_{i}}
\newcommand{\budgetagent}{\budgeti}

\newcommand{\indifferencebudget}{\budget_{I}}

\newcommand{\mech}{M}
\newcommand{\mechspot}{\mech^{\spot}}
\newcommand{\mechreserves}{\mech^{\reservation}}
\newcommand{\mechres}{\mechreserves}
\newcommand{\mechspotandres}{\mech^{\spot+\reservation}}

\newcommand{\mechspotandresbench}{B^{\spot+\reservation}}

\newcommand{\price}{p}
\newcommand{\pricespot}{\price_{\spot}}
\newcommand{\pricereservation}{\price_{\reservation}}

\newcommand{\util}{u}
\newcommand{\utili}{\util_{i}}

\newcommand{\utilreservation}{\util^{\reservation}}
\newcommand{\utilres}{\utilreservation}
\newcommand{\utilspot}{\util^{\spot}}
\newcommand{\utilagent}{\utili}
\newcommand{\utiliabove}{\utilagent^{*}}

\newcommand{\reserverCDF}{T}
\newcommand{\totalreserved}{T}

\newcommand{\fracreservedatprice}{Y}

\newcommand{\supply}{q}

\newcommand{\supplyCDF}{Q}

\newcommand{\eff}{EFF}
\newcommand{\welfare}{WEL}
\newcommand{\revenue}{REV}
\newcommand{\rev}{\revenue}

\newcommand{\welfarei}{w_i}

\begin{document}

\title{On-demand or Spot? Selling the cloud to risk-averse customers}
\author{Darrell Hoy\inst{1}\thanks{Part of this work was completed while D. Hoy was an intern at Microsoft Research.} \and
Nicole Immorlica\inst{2}\and
Brendan Lucier\inst{2}
}

\institute{
	University of Maryland\\
	\and
	Microsoft Research}

\maketitle

\begin{abstract}
In Amazon EC2, cloud resources are sold through a combination of an on-demand market, in which customers buy resources at a fixed price, and a spot market, in which customers bid for an uncertain supply of excess resources.  Standard market environments suggest that an optimal design uses just one type of market.  We show the prevalence of a dual market system can be explained by heterogeneous risk attitudes of customers.  In our stylized model, we consider unit demand risk-averse bidders.  We show the model admits a unique equilibrium, with higher revenue and higher welfare than using only spot markets. Furthermore, as risk aversion increases, the usage of the on-demand market increases.  We conclude that risk attitudes are an important factor in cloud resource allocation and should be incorporated into models of cloud markets. \end{abstract}

\section{Introduction}

Cloud computing allows clients to rent computing resources over the internet
to perform a variety of computing tasks, from hosting simple web servers to computing complex financial models.  By offloading these tasks to the cloud, clients avoid the necessity of procuring and maintaining expensive servers and infrastructure.  %
The current market leader in this industry is Amazon who launched its cloud platform, Amazon Elastic Compute Unit (EC2), in 2006.  Amazon uses its cloud internally for many of its own computations.  Additionally, Amazon contracts with large clients who reserve instances of cloud resources for long usage periods.  Due to natural variation in the nature of computing tasks from Amazon and its large clients, EC2 has a varying amount of leftover computing resources.  Amazon sells these resources to small clients. %

This leads to a natural question: how should a cloud provider price its resources to these small clients?  The pricing model adopted by Amazon has two main components: an on-demand market and a spot market.  In the on-demand market, clients may buy an instance of cloud resources at a fixed reservation price.\footnote{This might more naturally be called a ``reservation market'' and we switch to this terminology in the remainder of the paper; however we stick to the term ``on-demand'' for the current discussion as this is the term used by Amazon.}  After resources have been allocated internally, to large clients, and to clients in the on-demand market, 
extra supply might still remain.  This supply is sold in the spot market.  In a spot market, clients place bids for instances, and a price is set so that the available supply equals the total demand at that price.

Viewed through the lens of microeconomic theory, the persistence of this dual market is a curiosity at first glance.  Indeed, in standard economic environments, a risk-neutral, expected-utility-maximizing client who desires a cloud resource should simply buy it in whichever market is expected to have the lower price -- typically the spot market.  This suggests all sales should happen in the spot market, leaving the on-demand market defunct.  

That this is not reflective of reality stems from several factors.  Most apparent is that clients are rarely risk-neutral.  For example, it is easy to imagine that a company would have a soft budget set aside for computational costs.  They would then spend freely within the confines of this budget, and extend the budget cautiously when necessary to meet their computing needs.  This type of behavior suggests a tendency towards risk aversion on the part of the clients.  As the budget is freely available, clients might prefer to ``overspend'' to guarantee the required resources at the on-demand price.

We show in a stylized setting that the presence of heterogenous risk attitudes can explain the prevalence of a combined on-demand and spot market.  Specifically, this dual market induces a unique equilibrium in which more risk-averse customers (e.g., those with higher budgets) buy resources in the on-demand market and the others bid in the spot market.  We show that this equilibrium outcome outperforms the outcomes achieved by running only one of the two types of markets on its own in many key objectives.

\subsection{Results and Techniques}

In order to highlight the impact of risk aversion on the market, we focus our analysis on a simple setting in which there is only one type of computational resource being sold (e.g., a server with one core for one hour), and each buyer demands only a single instance of this cloud resource at any given time\footnote{Of course, this model abstracts away from many reasonable sources of risk aversion in the cloud, such as clients with diminishing marginal returns for multiple instances, the cost of prematurely terminating a long-running task.  Even ignoring these factors, our model still generates heterogeneous preferences toward on-demand versus spot pricing.
}. 
Formally, we assume a continuum of buyers, where the type of a buyer consists of a value for an instance and 
a utility curve that maps outcomes (i.e., allocation and price paid) to payoffs.  
The utility curve describes a buyer's attitude toward risk: for example, a buyer with a soft budget (as described above) would likely prefer to spend all of their budget all the time than to spend twice their budget half the time, and this preference is captured by a non-linear utility curve.
The buyer types are described by a joint common prior.
We assume the market is large; i.e., no single buyer has significant impact on the market outcome. 

The market works as follows.  First, the seller sets a price for on-demand instances.  This price should be high enough to guarantee that supply exceeds demand, motivated by the fact that resources are always available for purchase in on-demand markets in practice.  Buyers then realize their types (i.e., value/budget pairs) and choose whether to buy in the on-demand market.  After these decisions have been made, the unsold supply receives an exogenous shock, modeling variation in the demand of large clients.  Any remaining supply is then sold to the remaining buyers at a market-clearing price.\footnote{Since our model abstracts away from inter-temporal effects, we do not explicitly model the impact of fluctuating spot prices and changes to on-demand prices over time.  Investigating a repeated-game model of the market, and/or agents with time-dependent preferences (e.g., minimizing the cost of a large job subject to a deadline), is left as a direction for future work.}

We prove that this system has a unique (subgame-perfect) equilibrium for each choice of the on-demand market price.  We do this by analyzing 
the relationship between the spot price distribution and the distributions of clients' types and corresponding supply and demand.  It turns out that the distribution over spot prices up to a certain value $v$ depends only on the distributions of clients' types in the range $[0,v]$, and hence one can explicitly solve for the price distribution recursively.

This equilibrium satisfies a monotonicity property: agents 
that are more risk-averse
are more likely to purchase in the on-demand market, whereas agents 
that are less risk-averse 
are more likely to use the spot market.  Furthermore, as the distribution shifts such that agents become more averse to risk (in the sense of first-order stochastic dominance), 
we show more clients end up buying instances in the on-demand market and the revenue correspondingly increases.  This result is perhaps intuitive, but it is not obvious: as clients 
become more averse to risk,
they shift towards the on-demand market and hence both decrease supply and demand in the spot market.  This in turn could cause the spot price to shift either up or down, which would impact purchasing decisions of all clients.  By further arguing about the 
equilibrium of the market,
we show the shift towards the on-demand market in fact causes the spot price to increase thereby reinforcing the shift towards the on-demand market.  Therefore, the equilibrium is monotone with such shifts in the value and budget distribution.  This further illustrates the connection between the on-demand market and risk attitudes.

We leverage our equilibrium characterization to compare the dual market outcome to the outcome of a spot-only or on-demand-only market.  We are interested in the welfare, efficiency, and revenue properties of these markets.  The revenue of a market outcome is simply the sum of the payments, and is equal to the cloud provider's utility.  The welfare of an outcome is the total utility of all market participants including the cloud provider.  The efficiency is the total value of the cloud clients, ignoring payments.  In risk neutral environments, the welfare and efficiency are equal, but with risk-averse clients the welfare can be less than the efficiency.  

It is easy to see that a spot-only market is more efficient than a dual market, which in turn is more efficient than an on-demand-only market.  
This is because the spot market precisely generates the efficient outcome, even with exogenous supply uncertainty: in a spot market, allocation is monotone in value and thus a lower-valued client is never served in place of the higher-valued one.

Surprisingly, these efficiency comparisons do not extend to welfare.  As we show, the welfare of the dual market is better than the welfare of the spot market alone, regardless of the price set for the on-demand instances.  In particular, this is true even when the on-demand market price is set to maximize the revenue of the cloud provider.  This is not trivial: the on-demand market adds inefficiency, since clients with high value but 
low aversion to risk
may not wish to purchase on-demand, whereas clients with lower value but higher 
risk-aversion
might.  This leads to circumstances where lower-valued clients win but higher-valued clients lose.  However, since the clients 
that are winning in this scenario are actually more risk-averse,
the transfer of payments to the cloud provider increases welfare.  We show that the welfare increase due to additional transfers from risk-averse clients offset any inefficiencies in the allocation.  Moreover, since this welfare comparison holds at every setting of the on-demand price, it applies in particular to the price that maximizes revenue.  We show this price must also generate more revenue than a spot-only market, leading to increases in welfare and revenue.

In summary, a dual spot/on-demand market simultaneously improves both the revenue and welfare of a spot-only market.  We also show by example that while an on-demand market is revenue-optimal for risk-neutral buyers, a dual market can generate strictly higher revenue when buyers are risk-averse.  Furthermore, while a dual market may sometimes generate less revenue than an on-demand-only market, a dual market is always more efficient.  This suggests that a cloud provider, especially one that holds a dominant position in the marketplace, might prefer a dual market system.  This phenomenon is driven by heterogeneous risk attitudes that arise naturally in the context of cloud computation, leading us to posit that risk aversion is an important element to consider when one models the cloud marketplace.
\subsection{Related Work}

A number of papers explore cloud-computing market design. \citet{Zal13} consider designing a truthful auction where uncertainty lies in the arrival of demand and value profiles of bidders, whether they have a large job with deadline or general demand over time. \citet{ALIZ10} design a negotiation-based mechanism for setting price contracts in the presence of demand uncertainty.  \citet{BCCLN14} consider the pricing problem faced by a seller setting on-demand prices over multiple time periods and uncertain supply, with agents who arrive and have different deadlines for their tasks. The paradigm of a dual spot+reserve mechanism has also received a lot of attention.  \citet{WLL12} uses a Markov decision process to model the designer's choice of how to partition supply between the reserve and spot markets. \citet{AKK12} models the cloud market as a queuing model, in which a continuum of jobs arrive and have (private) waiting costs. They find that a fixed cost model provides greater expected revenue than a spot market. Additionally, recent works \citep{MSJ14,MH12} have focused on the problem faced by bidders in such a market: when to use the spot market and when to reserve. \citet{BBST13} analyzes the expected spot prices in comparison to their reservation prices, and find that it is very likely that Amazon is intentionally manipulating the price or supply distribution so as to provide users with more uncertainty in the spot market. 
In all of the models described above, agents are risk-neutral and do not have budgets. As far as we are aware, our work is the first to use risk aversion to explain the prevalence of a spot+reserve market.

Auctions for cloud computation resources share similarities to electricity markets, where the split between a spot market and a so-called ``futures" market is common.  Indeed, the use of both markets has been advocated to account for risk-averse buyers and sellers (see e.g., \cite{AC10}).  One difference is that, unlike the on-demand market for cloud computation, futures markets for electricity are typically resolved years in advance.

While most work in auction theory assumes risk-neutral agents, some work has been done for auctions with risk-averse bidders. Optimal auctions have been characterized for simple settings\ \citep{MR84, M83}, but the solutions are generally not expressible in closed form.  
It is therefore more common to study the simple auctions used in practice, with the general finding that second-price or spot-like auctions do poorly for revenue when compared with first-price auctions\ \citep{RS81, HMZ10, FHH13}. \citet{M87} has looked at the preferences of bidders in the auctions, and showed that first-price auctions not only can get more revenue than the second-price auction, but also can be preferred by bidders due to reduction in uncertainty around the payment. Our model of risk-aversion as the presence of a soft budget is closely related to the capacitated utilities model of \citet{FHH13}, where the capacity in their model corresponds to value minus budget in our model. They show that with capacitated agents, a simple first-price auction with reserve has revenue that approximates the revenue of the optimal mechanism.

Our results are of a similar flavor to the eBay-style buy-it-now auction considered by \citet{MK06}.  As in our model, they find that adding a buy-it-now option increases revenue, and as agents become more risk-averse, the optimal price increases.  Their model differs from ours in that they consider an explicitly randomized allocation rule, rather than clearing the market at a spot price, in order to incentivize use of the buy-it-now (i.e., reservation) option.

\section{Preliminaries}
\label{sec:prelim}

In our model, a single cloud provider (the seller) is selling %
compute resources to a continuum of clients (the bidders).  

\textit{Utility structure.}
Each bidder $i$ has value $\vali \in [0,1]$ for a single compute instance.  As is standard in the economics literature, we model the risk attitude of bidder $i$ through a utility function $\utili : \R_{\geq 0} \to \R_{\geq 0}$.  If the bidder obtains an instance and pays $p$, then her utility is $\utili( \vali - p )$.  We will assume that $\utili(0) = 0$, $\utili'(0)=1$, $\utili$ is continuous and non-decreasing, and that $\utili$ is not identically $0$.  Note that since $\utili(0) = 0$, we can think of $\vali$ as the maximum price at which bidder $i$ is willing to purchase an instance.  A bidder that does not obtain an instance will pay nothing and have utility $0$.  

We focus our attention on agents that are \emph{risk averse}.  That is, we will assume that utility curves are weakly concave, as is standard when modeling risk aversion.  We allow $\utili$ to be linear, in which case bidder $i$ is said to be risk-neutral.

Roughly speaking, an agent with a utility curve that is ``more concave'' will be more risk averse, in the sense that they are more likely to prefer guaranteed outcomes to uncertain lotteries.  More formally, we say that utility function $\util$ is more risk averse than $\util^*$, and write $\util \preceq \util^*$, if for every distribution $L$ over non-negative real values and every fixed value $d \geq 0$, if $\Ex[x \sim L]{\util(x)} \geq \util(d)$, then $\Ex[x \sim L]{\util^*(x)} \geq \util^*(d)$.  In other words, if an agent with utility curve $\util$ prefers a lottery $L$ over a guaranteed payout of $d$, then an agent with utility curve $\util^*$ would prefer the lottery as well.  This defines a partial order over utility curves.  Note that, under this definition, all (weakly) concave utility curves are (weakly) more risk-averse than a risk-neutral (i.e., linear) curve.  Note also that for twice-differentiable utility curves, $\util$ is more risk-averse than $\util^*$ if and only if the standard Arrow-Pratt measure of risk-averson is nowhere lower for curve $\util$ than for curve $\util^*$ \citep{Pratt64}.\footnote{We define risk aversion with respect to agent preferences directly, rather than via the Arrow-Pratt measure, to avoid requiring utility curves be twice differentiable.}

\textit{Demand structure}   Types are distributed according to a joint distribution $\typeCDF$ on pairs $(\val, \util)$.
For ease of exposition, we will assume throughout that $\typeCDF$ is supported on a finite collection of $(\val, \util)$ pairs.  Write $V$ and $U$ for the (finite) sets of values and utility curves that support $\typeCDF$, and for $(\val, \util) \in V \times U$ we will write $\typePDF(\val, \util)$ for the probability that an agent has type $(\val, \util)$.

We will use $\valCDF(\val)$ to refer to the induced distribution over values; that is, $\valCDF(\val)$ is the probability that an agent's value is at most $\val$.
We will assume a large-market condition, which is that the aggregate demand is distributed exactly according to the type distribution $\typeCDF$.

\textit{Supply structure} The supply of instances, $\supply$, is unknown to the bidders and seller until the instances are to be allocated.  The supply is then drawn from a distribution, $\supplyCDF$.  We will normalize the supply so that $\supply$ represents the fraction of the market that can be simultaneously served, hence $\supply\in [0,1]$. 

\subsection{Auction Formats}
\label{sec.auctions}

We will consider three auction formats in this paper: spot auctions, reservation auctions (previously referred to as on-demand), and dual (or spot+reservation) auctions.

\begin{description}
\item[Spot ($\mechspot$)]
One type of auction to run is a \emph{market-clearing} auction, or a spot auction.  In this auction, buyers submit bids.  A market-clearing price $\pricespot$ is chosen such that the quantity of bids exceeding $\pricespot$ is equal to the supply. Under our unit-demand and large market assumptions, it is a dominant strategy for a bidder to bid her value; henceforth we assume the bids in the spot market equal the values.
We observe that a market-clearing price always exists in our market, even in the presence of 
non-linear utilities:
with available supply $\supply$, and distribution over values $\valCDF$, the market clearing price, written $\pricespot(\supply)$, is precisely\footnote{This price may not be unique if $q = 0$ or $q=1$.  In these cases we define $\pricespot(\supply)$ to be the supremum of prices satisfying the written condition, which will be $\infty$ for $q = 0$.} the value for which $\supply = 1-\valCDF(\pricespot(\supply))$.

\item[Reservation ($\mechres$)] In a reservation-only (or ``on-demand") market, the auctioneer sets a fixed price $\pricereservation$ per instance, in advance of seeing the realization of supply.  Price $\pricereservation$ need not be a market-clearing price.  If there is not enough supply to satisfy the demand for instances at this price, the winning bidders are chosen uniformly at random from among those who wish to purchase.

\item[Spot+Reservation ($\mechspotandres$)] In a spot and reservation market, the auctioneer first sets a fixed price $\pricereservation$ and runs a reservation auction.  The remaining inventory of supply (if any) is then sold via a spot auction.
The exact timeline of events in the spot and reservation auction $\mechspotandres$ is as follows:

\begin{enumerate}
\item Auctioneer announces reservation price $\pricereservation$.
\item Bidders realize types 
$(\valagent, \utilagent) \sim \typeCDF(\val, \util)$.
\item Each bidder decides whether to purchase an instance in the reservation auction, indicated by $\actionreservestage\in \{0,1\}$. Let 
$\totalreserved = \sum_{v,u} \actionreservestage(\val, \util)\ \typePDF(\val, \util)$
be the total volume of reserved instances. 
\item Auctioneer realizes supply $\supply\sim \supplyCDF$, and reserved instances are allocated as in the reservation market described above.  
\item If $\supply > \totalreserved$, the auctioneer runs a market-clearing auction to clear the excess capacity. Let $\pricespot(\supply)$ be the resulting market-clearing spot price.
\end{enumerate}

Note that our specification does not ask bidders to decide whether or not to participate in the spot market. The fact that bidders are unit demand, and that the spot auction is truthful under our large market assumption, implies that in equilibrium bidders will bid (truthfully) in the spot auction if (and only if) they don't buy an instance in the reservation auction.

\end{description}

For a given (implicit) strategy profile for mechanism $\mechspotandres$, we will write $\spotCDF(\pricespot)$ for the cumulative distribution function of the resulting spot prices.

\textit{Solution concept: subgame-perfect equilibrium.}
For each of these auctions, the solution concept we apply is \emph{subgame-perfect equilibrium}.  A strategy profile for a multi-stage game forms a subgame-perfect equilibrium (SPE) if, at every stage $t$ of the game and every possible history of actions by players in previous stages, no agent can benefit by unilaterally deviating from her prescribed strategy from stage $t$ onward.

For the spot auction and reservation auction, there is only one stage of the resulting game and hence equilibria are straightforward: each agent chooses to purchase her utility-maximizing quantity of instances given the specified price.  

For mechanism $\mechspotandres$, we can characterize the SPE as follows.
In the second (i.e., spot) stage of the mechanism, the equilibrium condition implies that agents always purchase instances if and only if their value is above the realized spot price.  Thus, the only strategic choice to be made by agents is in the first (i.e., reservation) stage of the mechanism, where each agent must select whether to purchase an instance in the reservation market.  We will therefore define a strategy profile $\strats$ to be a mapping from a type $(\val, \util)$ to an action $\{0,1\}$, where $\strats(\val,\util)$ is interpreted as the number of instances to purchase in the reservation market.  Note that the distribution over market-clearing prices in the second stage is completely determined by the actions of agents in the first stage, and hence is determined by $\strats$.  An equilibrium is then a strategy profile such that no agent can benefit by unilaterally deviating from strategy $\strats$ (i.e., by reserving more or fewer instances), given the distribution of spot prices implied by $\strats$.

\subsection{Objectives}
We consider three objectives when evaluating mechanisms: revenue ($\rev$), welfare ($\welfare$) and efficiency ($\eff$). The revenue of a mechanism $\mech$, $\rev(\mech)$, is the sum of the payments made to the auctioneer. 
Note that for the spot+reservation mechanism,
$
\rev(\mechspotandres) = \pricereservation\totalreserved +  \Ex[\supply\sim\supplyCDF]{ \pricespot(\supply)(\supply-\totalreserved)}
$
where we used the fact that $T \leq q$ with probability $1$.
The welfare of a mechanism $\welfare(\mech)$ is the sum of utilities of all agents, including the auctioneer (whose utility is precisely the revenue of the mechanism). 
The efficiency $\eff(\mech)$ of a mechanism measures the value created, without considering the welfare lost due to the (non-linear) utility functions of agents. For the spot+reservation mechanism:
$
\eff(\mechspotandres) = \Ex[(\valagent, \utilagent) \sim \typeCDF]{\valagent \cdot \strat(\valagent,\utilagent) + \valagent \cdot (1 - \strat(\valagent,\utilagent))\cdot \spotCDF(\valagent)} 
$

If an agent reserves, her value generated is $\valagent$.  If she does not reserve, her value generated is $\valagent \cdot \spotCDF(\valagent)$, which is her value times the probability that the spot price is below her value. Note that if agents are risk-neutral (i.e., have the identity function as their utility functions), then $\eff(\mech)=\welfare(\mech)$. 

\section{Equilibrium Behavior \& Analysis}

In this section, we analyze the choices of bidders and use this to characterize equilibrium of the spot+reservation market. We begin by noting the relationship between the spot price distribution and the distributions of supply, type, and reservation demand in equilibrium.  Recall that $\supplyCDF$ denotes the CDF of the supply distribution.  

\begin{lemma}\label{lem:generalchar}
Fix strategy profile $\strats$, let $\spotCDF$ be the distribution of spot prices under $\strats$, and let $\reserverCDF(\price)$ 
be the volume of reserved instances demanded from agents with value at most $\price$, under $\strats$. 
Then 
$\strats$ forms an equilibrium if and only if, for all $p$, 
\begin{align}
\spotCDF(\price) &= 1-\supplyCDF\left(1-\typeCDF(\price) + \reserverCDF(\price)\right), \text{ and }\\
T(p) &= \sum_{\substack{v \in V\\v \leq p}} \sum_{u \in U} f(v,u) \cdot \mathbbm{1}\left[ u(v - p_r) \geq \Ex[p \sim S]{u(\max\{v-p, 0\})} \right].
\end{align}
\end{lemma}
\begin{proof}
The probability that the spot price is at most $\price$ is exactly the probability that the supply is greater than necessary to satisfy all of the demand for resources from bidders with higher marginal values than $\price$, %
plus all reservation demand for resources with lower marginal values than $\price$.  On the other hand, the volume of reserved resources demanded from agents with value at most $p$, at equilibrium, is precisely the probability that such an agent will prefer the deterministic reservation outcome to the lottery over outcomes determined by the distribution over spot market prices.
\end{proof}

In light of Lemma \ref{lem:generalchar}, we will tend to equate equilibria with the resulting distributions $\spotCDF$ and $\reserverCDF$, rather than with an explicit strategy profile $\strats$.

\begin{lemma}\label{lem:unit-monotone}
Purchasing in the reservation stage is monotone in the risk-aversion of
$\utilagent$: if a $(\valagent, \utilagent)$ bidder (weakly) prefers to reserve an instance, then a 
$(\valagent, \utiliabove)$ 
bidder with %
$\utiliabove \preceq \utilagent$ (weakly) prefers reserving. 
\end{lemma}
\begin{proof}
We begin by considering the special event in which the agent is not allocated an instance even if they reserve, due to the supply being insufficient to honor all reservations and the agent not being selected randomly as a winner.  If this event occurs, the bidder's utility will necessarily be $0$, and this is independent of their utility curve and their chosen action (since, if $q < T$, no agent that enters the spot market will obtain an instance).
It therefore suffices to consider the agent's expected utility conditional on the event that the agent will be allocated an instance with certainty if they choose to reserve.  With this in mind, the utilities of a unit demand agent from reserving or participating only in the spot market, respectively, are
\begin{align*}
	\utilreservation(\valagent, \utilagent) &= \utilagent(\valagent - \pricereservation), \\
	\utilspot(\valagent, \utilagent) &= \Ex[\pricespot\sim\spotCDF]{\utilagent(\valagent - \pricespot) \cdot \Indic{\pricespot\leq \valagent} } .
\end{align*}
We now want to show that $\utilreservation(\valagent, \utilagent) \geq \utilspot(\valagent, \utilagent)$ implies 
$\utilreservation(\valagent, \utiliabove) \geq  \utilspot(\valagent, \utiliabove)$.  
Note that, fixing $v_i$, the spot market generates a certain lottery $L$ over values $(\vali - \pricespot)$, and the reservation market generates a certain value $\vali - \pricereservation$.  Thus, from the definition of risk aversion, if an agent with utility curve $\utili$ prefers the certain outcome to the lottery $L$, and $\utiliabove \succeq \utili$, then an agent with utility curve $\utiliabove$ prefers the certain outcome as well.
\end{proof}

\subsection{Equilibrium Existence and Uniqueness}
\label{sec:3.1}
We are now ready to establish uniqueness of equilibrium.  One subtlety about equilibrium uniqueness is the manner in which buyers break ties.  If a positive mass of agents is indifferent between the spot and reservation markets, there may be multiple market outcomes consistent with those preferences.  We will therefore fix some arbitrary manner in which bidders break ties, which could be randomized and heterogeneous across bidders.  Our claim is that for any such tie-breaking rule, the resulting equilibrium will be unique.

\begin{lemma}\label{lem:unique}
There is a unique equilibrium of  $\mechspotandres$.  Moreover, this equilibrium is computable in time proportional to the size of the support of type distribution $\typePDF$.
\end{lemma}

\begin{proof}
As shown in Lemma~\ref{lem:generalchar}, the challenge of characterizing equilibrium essentially reduces to characterizing the fraction of bidders who reserve at a given price, $\reserverCDF(\price)$.  This is because $\reserverCDF$ determines the distribution $S$ over spot prices, and $S$ (together with an arbitrary tie-breaking rule) uniquely determines the strategy profile $s$, since this can be inferred from the expected utility when choosing the spot market.  Thus, to show uniqueness and existence of equilibrium, it suffices to show uniqueness of the functions $S$ and $T$.

We will prove that, for all $v \in V$, $T(v)$ and $S(v)$ are uniquely determined by the functions $T$ and $S$ restricted to values less than $v$.  The result will then follow by induction on the elements of $V$.

Consider first an agent with value $v = \min V$.  Recall that the spot price is always at least $v$.  Thus, if $\pricereservation < v$ then the agent will always reserve, if $\pricereservation > v$ then the agent will always choose the spot market, and if $\pricereservation = v$ the agent will be indifferent and apply the fixed tie-breaking rule.  In each case, the value of $T(v)$ is uniquely determined, and thus $S(v)$ is as well.

Now choose $v > \min V$, and suppose $T$ and $S$ are determined for all elements of $V \cap [0, v)$.  We claim the distribution of the random variable $\max\{v - p_s, 0\}$, where $p_s$ is distributed according to $S$, is then uniquely determined.  This is because the non-zero values of this random variable are distributed according to $S$ restricted to values in $V \cap [0,v)$.  But, by Lemma~\ref{lem:generalchar}, this random variable determines the value of $T(v)$, which in turn determines the value of $S(v)$.  Thus $T(v)$ and $S(v)$ are uniquely determined by $S$ and $T$ on $V \cap [0,v)$, as required.  Moreover, they can be explicitly computed by evaluating the summation in Lemma~\ref{lem:generalchar}. 

\end{proof}

\subsection{An Example: Soft Budgets}
\label{sec:budgets}

In this section we present a special case of risk-aversion, driven by soft budgets, and give an interpretation of our equilibrium characterization for this case. 

Suppose that each buyer $i$ is characterized by their value $\vali$ for a compute instance and a soft budget $\budgeti \in [0, \vali]$.  We think of $\budgeti$ as a budget of funds that has been allocated to acquiring a compute instance.  If the buyer obtains an instance but pays less than $\budgeti$, the residual budget is lost: it is as if they had paid $\budgeti$.  On the other hand, if the buyer pays more than $\budgeti$, they suffer no additional penalty; they simply incur the cost of their payment.

This scenario is captured by a piecewise linear utility curve, $\utili(z) =  \min(z, \vali - \budgeti)$. 
Note that if a buyer with budget $b_i$ obtains an instance and pays $p_i  > \budgeti$ their utility is $\utili(\vali - p_i) = \vali - p_i$, whereas  if they pay $p_i < \budgeti$ their utility is $\utili(\vali - p_i) = \vali - \budgeti$.
For a fixed value $\vali$, an agent with a higher budget is more risk-averse.  To see this, note that if a buyer {\it strictly} prefers a lottery $L$ to a deterministic outcome $d$, then it must be that $d < \vali - \budgeti$ (since otherwise $d$ must be utility-maximizing).  In this case, decreasing the budget can only make the lottery more valuable, while not affecting the utility from the deterministic outcome.  Thus, a decreased budget can only increase the propensity to select a lottery over a deterministic outcome.

The monotonicity result in Lemma~\ref{lem:unit-monotone} thus results in a partitioning of agents that prefer the spot market to the reservation market and vice versa by an indifference curve over budgets. See Figure~\ref{fig:partitioning} for an illustration. Any agent with $(\vali, \budgeti)$ below the curve prefers the spot market; any agent above the curve (such that $\valagent\geq \budgetagent$), prefers the reservation market.   Lemma~\ref{lem:unique} shows that this indifference curve is unique, given the distribution over agent types, and precisely specifies which agents choose to reserve and which enter the spot market. 

\begin{figure}[t]
\centering
\begin{tikzpicture}[xscale=3.1, yscale=2.1, domain=0:1, smooth]
	\draw [ultra thick]  (0.5,0.5)  .. controls ( 0.6,0.45) .. (0.7,0.18) .. controls ( 0.8,0.05) .. (1,0.05);
    \node at (1,1) {};
	\draw[-] (0,0) -- (1,1);
	\node at (0.7,-0.09) {value ($\valagent$)};
	\node[rotate=90] at (-0.09, 0.7) {budget ($\budgetagent$)};    
	\node at (0.85,0.5) {reserve};
	\node at (0.4,0.7) {$\emptyset$};		
	\node at (0.5,0.2) {spot};	    
	\node at (0.85, 0.25) {$\indifferencebudget(\val)$};
    \draw[-] (0,0) -- (0,1);
	\draw[-] (0,0) -- (1,0);
\end{tikzpicture}
\caption{With soft-budgets, a monotone-decreasing indifference curve partitions agents into those that reserve an instance and those who rely on the spot market.  %
 \label{fig:partitioning}}
\end{figure}
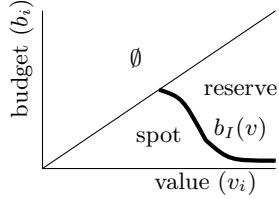

\section{Comparative Statics}
\label{sec:comparative}

In this section, we first consider the impact of changes to buyer risk attitudes.  We show that as agents become more risk averse, more agents use the reservation market and revenue increases, for every setting of the reservation price.  Second, we compare the reservation+spot mechanism to the spot market and the reservation market.  We first show that the combination mechanism's outcomes are more efficient than running a reserve market on its own.  We then show that it generates both more revenue and more welfare than running only a spot market. Proofs are included in the Appendix.

The results in this section hold under two assumptions on the reservation price set by the seller.
First, we will make the assumption that $\pricereservation$ is set high enough so that, in the resulting equilibrium, $\Pr_{\supply \sim \supplyCDF}[ \supply < T ] = 0$.  That is, over the uncertainty in supply, the mechanism can serve the reserved instances with certainty.  This assumption is motivated by the fact that these instances are typically viewed as guaranteed by the mechanism.

Another natural and practical property is that the reservation price $\pricereservation$ be set high enough that it will be greater than the expected spot price.  That is, $\pricereservation$ is large enough that it is more costly, in expectation, to reserve a guaranteed instance than to bid for an instance in the spot market.  %

\subsection{Effect of Increased Risk Aversion}
\label{sec:increase.budgets}

\newcommand{\alt}[1]{\overline{#1}}

We consider the impact of an increase in risk aversion.  Consider type distribution $F$ and a type distribution $F^+$ induced by a pointwise transformation $g^+(U,V)\to (U,V)$ applied to each point in $F$ which weakly increases risk aversion and does not affect values. Specifically, for any  $(\util^+, \val^+)=g^+(\util, \val)$, $\val^+=\val$ and $\util^+ \preceq \util$. %
In the following lemma, we show that such a change can only increase the fraction of agents who choose to reserve, and can only increase revenue. 

\begin{lemma} \label{lem:moreRAmorerev}
For mechanism $\mechspotandres$, and for any reserve price $\pricereservation$, if risk aversion of agents increases then the fraction of agents who purchase in the reservation stage increases, as does the expected revenue of the mechanism.
\end{lemma}

The intuition underlying Lemma \ref{lem:moreRAmorerev} is as follows.  The first order effect from a change in risk aversion is an increase in $T$, the fraction of users who choose to reserve at a given price.  This increase in reservations translates into higher spot prices, since it reduces the quantity sold in the spot market.  Higher spot prices in turn cause more users to prefer to reserve, which can only increase spot prices further.  This can be shown by induction over agent values.

\subsection{Comparing Mechanisms}
\label{sec:comparing}

We now compare welfare and revenue of $\mechspotandres$ to $\mechspot$ and $\mechres$. Here we make use of the two assumptions discussed in Section \ref{sec.auctions}: first, the reservation phase is not oversubscribed, i.e., the reservation price is set sufficiently high that there will be sufficient supply to fulfill the demand for reserved instances; and second, the reservation price is sufficiently high to be above the expected spot price.

We begin by comparing the revenue of the dual mechanism with the expected revenue of a spot-only market.  Note that, trivially, the best revenue of the combined mechanism is at least the revenue of a spot market; this is because, in the combined mechanism, the reserve price can be set sufficiently high that all customers buy in the spot market.  We show something stronger: for \emph{every} choice of reservation price, the revenue of the combined mechanism is at least that of a spot market run in isolation. 

\begin{lemma}
For any choice of the reservation price satisfying our assumptions, 
the expected revenue of the reservation and spot mechanism is weakly greater than the revenue of the spot-only market.
\end{lemma}

\begin{proof}[sketch]
As in Lemma~\ref{lem:moreRAmorerev}, as risk aversion increases, revenue increases for a fixed reservation price. 
Fix a reservation price and consider starting with a distribution of risk-neutral agents. 
These agents will all bid in the spot market and thus the outcome (and in particular, the revenue) will be identical to the spot-only mechanism. By deforming the utility curves of the agents in a manner that only increases risk aversion, until they match the correct distribution, and applying Lemma~\ref{lem:moreRAmorerev}, we can conclude that the revenue of the dual market only increases while the revenue of the spot-only mechanism, which is unaffected by the utility curves, %
remains the same. 
\end{proof}

\begin{example}
	This example illustrates the revenue of the dual mechanism can be strictly greater than both the spot and the reservation mechanisms. We consider an example in which agent utility curves are specified by soft budgets, as in Section~\ref{sec:budgets}.  Recall that a budget of $0$ corresponds to risk-neutrality.  Take $\epsilon > 0$ to be sufficiently small, and consider the following distribution over buyer types:
\begin{itemize}
\item with probability $0.5 - \epsilon$, $(v,b) = (5,0)$
\item with probability $0.5$, $(v,b) = (10, 10-\epsilon)$
\item with probability $\epsilon$, $(v,b) = (20, 0)$
\end{itemize}
The supply is distributed such that $q = 1-\epsilon$ with probability $0.8$, and otherwise $q = 0.5 + \epsilon/2$. In the spot-only auction, the spot price is $5$ with supply $q=1-\epsilon$ and $10$ with supply $q=0.5+\epsilon/2$, giving revenue of $5-3\epsilon$. %
In the reservation market, the optimal reserve price is $10$, which generates a total revenue of $5 + 9\epsilon$.%

Consider the dual mechanism with reservation price $10-\epsilon$.  At equilibrium, the buyers of type $(10,10-\epsilon)$ strictly prefer to reserve (obtaining utility $\epsilon$ with probability $1$, rather than utility $\epsilon$ with probability $0.8$), whereas the buyers of type $(20,0)$ strictly prefer to participate in the spot market (obtaining utility $(20-5)$ with probability $0.8$, rather than utility $10+\epsilon$ with probability $1$).  
The revenue is then $7 - \tfrac{5}{2}\epsilon$, greater than spot-only revenue $5-3\epsilon$, and reservation-only revenue $5 + 9\epsilon$ for sufficiently small $\epsilon$.
\end{example}

\begin{lemma}
For any choice of the reserve price satisfying the assumptions above,
the expected efficiency of the reservation+spot mechanism is weakly greater than the efficiency of the reservation market with the same reservation price.
\end{lemma}
\begin{proof}
	Recall that the efficiency of a mechanism is the expected value generated by the agents, ignoring the welfare lost due to the nonlinear utility functions. For any realized supply $\supply$ then, weakly more people are served in $\mechspotandres$, hence efficiency is greater. 
\end{proof}

\begin{theorem} \label{thm:welfarebound}
In any equilibrium of the the spot and reservation mechanism where the reservation price is set above the expected spot price, the expected welfare of the reservation and spot mechanism is weakly greater than the expected welfare of the spot-only mechanism.
\end{theorem}

\begin{proof}[sketch]
Note that, relative to a spot market, introducing a reservation price adds inefficiency.  This is because if a bidder is willing to reserve to get a guaranteed instance, any time they would not have one in the spot market, there is a higher valued bidder than them who could be allocated.
	
However, welfare is increased when a bidder chooses to reserve. Consider an agent with value $\valagent$ who is willing to reserve at price $\pricereservation$. Reserving increases his utility and the auctioneer is receiving more revenue, because the reservation price is greater than the expected spot price (by assumption). 

The full proof consists of three parts. First, we define a benchmark $\mechspotandresbench$ that is just like $\mechspotandres$ except agents who reserve pay the spot price instead of the reservation price. We then show that the welfare of $\mechspotandres$ is greater than the welfare of $\mechspotandresbench$, which follows largely because the expected reservation price is greater than the expected spot price. Finally, we show that spot prices increase when agents choose to reserve, which leads to the welfare of $\mechspotandresbench$ being greater than the welfare of $\mechspot$, and hence the welfare of $\mechspotandres$ is greater than $\mechspot$.
\end{proof}

\bibliography{bibs}{}
\bibliographystyle{apalike}

\appendix
\section{Appendix: Proofs}

\subsection{Proof of Lemma~\ref{lem:moreRAmorerev}}

\begin{proof}
Let $F$ be the original distribution over types, and let $\alt{F}$ denote a distribution in which the risk aversion of agents increases, for each fixed value $v$. %
Let $\alt{T}$ and $\alt{S}$ refer to the volume of reserved instances and distribution of spot prices under the alternate distribution $\alt{F}$.  As before, $\indifferencebudget$, $T$, and $S$ correspond to the respective quantities for the original distribution $F$.

We will show that for any price $\price\geq \pricereservation$, it is the case that $\alt{T}(\price) \geq T(\price)$ and $\alt{S}(\price) \leq S(\price)$.  To see how this implies the lemma, note that these inequalities imply that the increase in risk aversion of agents results in (a) more agents reserving, and (b) higher spot prices.  Because the reservation price is greater than the expected spot price, the now-reserving agents pay more than they were paying before and total revenue increases.

We will prove the desired inequalities by induction on $\price \in V$.  Note first that for any $\price < \pricereservation$, an agent with value $\price$ will certainly choose to purchase in the spot market regardless of their risk attitudes, and hence $\alt{T}(\price) = T(\price)$ and $\alt{S}(\price) = S(\price)$.  For any $\price \geq \pricereservation$, recall that
\[ T(p) = \sum_{\substack{v \in V\\v \leq p}} \sum_{u \in U} f(v,u) \cdot \mathbbm{1}\left[ u(v - p_r) \geq \Ex[p \sim S]{u(\max\{v-p, 0\})} \right] \]
and
\[ \alt{T}(p) = \sum_{\substack{v \in V\\v \leq p}} \sum_{u \in U} \alt{f}(v,u) \cdot \mathbbm{1}\left[ u(v - p_r) \geq \Ex[p \sim \alt{S}]{u(\max\{v-p, 0\})} \right]. \]
As in Lemma~\ref{lem:unique}, the dependence on $S$ and $\alt{S}$, respectively, is limited to values in $V \cap [0,v)$.  Since we assume inductively that $S(p) \geq \alt{S}(p)$ for all $p \in V \cap [0,v)$, the distribution over spot market utilities under $\alt{F}$ is stochastically dominated by the distribution under $F$, for agents with value $v$.  Since $\alt{F}$ additionally increases risk aversion relative to $F$, fixing $v$, we conclude that weakly fewer agents with value $v$ would choose to reserve under $\alt{F}$ relative to $F$.  That is, $\alt{T}(\price) \geq T(\price)$.  Then, since
\[ \spotCDF(\price) = 1-\supplyCDF\left(1-\typeCDF(\price) + \reserverCDF(\price)\right) \]
and
\[ \alt{\spotCDF}(\price) = 1-\supplyCDF\left(1-\typeCDF(\price) + \alt{\reserverCDF}(\price)\right) \]
we can immediately conclude $\alt{S}(\price) \leq S(\price)$ as well.  The desired result then follows by induction.
\end{proof}

\subsection{Proof of Theorem~\ref{thm:welfarebound}}

We now proceed with the proof of Theorem \ref{thm:welfarebound}: that the welfare of the spot+reservation mechanism is greater than the welfare of the spot-only mechanism.
Let $\welfare_{\supply}(\mech)$ be the welfare of mechanism $\mech$ at fixed supply $\supply$. Then we can write the total welfare of $\mech$ as
\begin{equation}
\welfare(\mech) = \Ex[\supply\sim\supplyCDF]{\welfare_{\supply}(\mech)}
\end{equation}

Write $\welfarei^{\vali}(\price)$ for the welfare generated by an agent paying price $\price$, hence $\welfarei^{\vali}(\price) = \utili(\vali - \price) + \price$. Note that $\welfarei(\price)$ is monotone non-decreasing in $\price$, since $u_i$ is a concave non-decreasing function with $u_i'(0)=1$.  %

Let $\pricespot^{s}(\supply)$ be the price in the spot-only mechanism $\mechspot$ if the realization of supply is $\supply$, and $\pricespot^{s+r}(\supply)$ the spot price in spot and reservation mechanism $\mechspotandres$. Note that $\pricespot^{s+r}(\supply)\geq \pricespot^{s}(\supply)$ for all $\supply$, no matter how the buyers behave. This is because $\pricespot^{s}(\supply)$ is precisely the minimal price at which a $\supply$ fraction of the buyers will purchase, and hence if $\pricespot^{s+r}(\supply) < \pricespot^{s}(\supply)$ then more than a $\supply$ fraction of buyers must be purchasing in $\mechspotandres$ (since they would buy in the spot market if they didn't reserve), which is impossible.

We now analyze the welfare of $\mechspotandres$, breaking down the welfare generated by those bidders who buy in the spot market and who reserve. Let $\fracreservedatprice(v)$ be the fraction of bidders with value $v$ who reserve.  Then
\begin{align}
\welfare_{\supply}(\mechspotandres) =& \sum_{v\leq \pricespot^{s+r}(\supply)} \welfarei^{v}(\pricereservation) \fracreservedatprice(v) \typePDF(v)\\
+& \sum_{v > \pricespot^{s+r}(\supply) } \left( \welfarei^{v}(\pricespot^{s+r}(\supply)) (1-\fracreservedatprice(\val))+\welfarei^{v}(\pricereservation) \fracreservedatprice(\val) \right) \typePDF(v).
\end{align}

Note the split between the two summations at $\pricespot^{s+r}(\supply)$. Agents with values below $\pricespot^{s+r}(\supply)$ are served only if they reserve, whereas agents with values above $\pricespot^{s+r}(\supply)$ are served whether or not they reserve --- the only question for their welfare is whether they pay $\pricereservation$ or $\pricespot^{s+r}$. Call each of these summations $\welfare_{\supply}^{-}(\mechspotandres)$ and $\welfare_{\supply}^{+}(\mechspotandres)$ respectively.

We now consider the welfare from a benchmark, $\mechspotandresbench$. For a given supply $q$, if agents reserve and have values above $\pricespot^{s+r}(\supply)$, we assume that they pay the spot price $\pricespot^{s+r}(\supply)$ instead of the reservation price $\pricereservation$. For agents who reserve with values below the spot price $\pricespot^{s+r}(\supply)$, we assume they pay the spot price $\pricespot^{s+r}(\supply)$ and (magically) get no utility or disutility from doing so, hence the only welfare generated is the welfare the designer experiences from the payment: $\welfarei^{z}(\pricespot^{s+r}(\supply)) = \pricespot^{s+r}(\supply)$. 

As we did for the combined mechanism, we can write the welfare of the benchmark as 
$$\welfare_{\supply}(\mechspotandresbench) = \welfare_{\supply}^{-}(\mechspotandresbench) + \welfare_{\supply}^{+}(\mechspotandresbench),$$ where $\welfare_{\supply}^{-}(\mechspotandresbench)$ denotes the contribution to welfare from bidders with values below $\pricespot^{s+r}(\supply)$, and $\welfare_{\supply}^{-}(\mechspotandresbench)$ is the contribution to welfare from bidders with values above $\pricespot^{s+r}(\supply)$.  For bidders with values below the spot price, we have
\begin{align*}
\welfare_{\supply}^{-}(\mechspotandresbench) =& \sum_{\val \leq \pricespot^{s+r}(\supply)} \pricespot^{s+r}(\supply) \fracreservedatprice(\val) \typePDF(\val)\\
=& \pricespot^{s+r}(\supply) \totalreserved(\pricespot^{s+r}(\supply))
\end{align*}

The last line followed because $\totalreserved(\pricespot^{s+r}(\supply))$ is exactly the volume of agents with value at most $\pricespot^{s+r}(\supply)$, hence $\totalreserved(\pricespot^{s+r}(\supply)) =  \sum_{v\leq \pricespot^{s+r}(\supply)} \fracreservedatprice(v) \typePDF(v)$. For the bidders with values above the spot price, the welfare satisfies
\begin{align*}
\welfare_{\supply}^{+}(\mechspotandresbench) =& \sum_{v>\pricespot^{s+r}(\supply)} \welfarei^{v}(\pricespot^{s+r}(\supply))\typePDF(v).
\end{align*}

Whether or not the benchmark welfare for agents with values above the spot price is above or below the actual welfare depends on whether the spot price is above or below the reservation price.  But, as the following claim shows, the welfare benchmark will be less in expectation than the actual welfare: %

\begin{claim}
\label{claim:1}
\begin{equation}
\Ex[\supply\sim\supplyCDF]{\welfare_{\supply}(\mechspotandres)} \geq \Ex[\supply\sim\supplyCDF]{\welfare_{\supply}(\mechspotandresbench)}
\end{equation}
\end{claim}
\begin{proof}
The expected payment from every agent who reserves in $\mechspotandresbench$ is $\Ex[\supply\sim\supplyCDF]{\pricespot^{s+r}(\supply)}$, which by assumption satisfies $\Ex[\supply\sim\supplyCDF]{\pricespot^{s+r}(\supply)}\leq \pricereservation$. Thus the total revenue in the mechanism $\mechspotandres$ is greater than in the benchmark $\mechspotandresbench$.

For all agents who do reserve, we know their utility from reserving is more than their utility if they had not reserved and only participated in the spot market, $\utilres(\valagent, \utilagent) \geq \utilspot(\valagent, \utilagent)$. The utilities of all agents who do not reserve are the same in both, so they are indifferent.

Summing the utilities of all the agents and the revenue of the designer gives 
\begin{equation*}
\Ex[\supply\sim\supplyCDF]{\welfare_{\supply}(\mechspotandres)} \geq \Ex[\supply\sim\supplyCDF]{\welfare_{\supply}(\mechspotandresbench)}.
\end{equation*}
\end{proof}

We now argue that the welfare of the benchmark is an upper bound on the welfare from the spot-only mechanism.

\begin{claim}
\label{claim:2}
\begin{equation}
\welfare_{\supply}(\mechspotandresbench)\geq \welfare_{\supply}(\mechspot) 
\end{equation}
\end{claim}
\begin{proof}

The agents who purchase in $\mechspot$ are precisely those with values above $\pricespot^{s}(\supply)$, the spot price for the spot-only mechanism.  Consider separately the agents with values above and below the spot price $\pricespot^{s+r}(\supply)$ for the combined mechanism. 

Agents with values above $\pricespot^{s+r}(\supply)$ are always allocated in the benchmark and the spot-only mechanism. As the spot price (and hence benchmark payment) is higher in the spot+reservation mechanism than the spot-only mechanism for a given supply, and $\welfarei$ is non-decreasing in $\price$, $\welfarei(\pricespot^{s+r}(\supply))\geq \welfarei(\pricespot^{s}(\supply))$, thus
\begin{align}
\sum_{v > \pricespot^{s+r}(\supply) } \welfarei^{v}(\pricespot^{s}(\supply)) \typePDF(v) &\leq \sum_{v> \pricespot^{s+r}(\supply)}  \welfarei^{v}\pricespot^{s+r}(\supply) \typePDF(v) \nonumber\\
&= \welfare_{\supply}^{+}(\mechspotandresbench). \label{eq:valabove} 
\end{align}

Consider now agents with values below $\pricespot^{s+r}(\supply)$. If the agent does not reserve in the spot+reservation mechanism, then we know that the price paid in the benchmark is higher than the welfare generated from receiving the item: $\valagent \leq \pricespot^{s+r}(\supply)$, hence we know that $\welfarei^{z}(\pricespot^{s}(\supply)) \leq \pricespot^{s+r}(\supply)$. Recall that  $\totalreserved(\pricespot^{s+r}(\supply))$ is the volume of agents with values below $\pricespot^{s+r}(\supply)$ who reserve. Thus, 
\begin{align}
\sum_{v | \pricespot^{s}(\supply) < v \leq \pricespot^{s+r}(\supply)} \welfarei^{z}(\pricespot^{s}(\supply)) \typePDF(v) &\leq \sum_{v | \pricespot^{s}(\supply) < v \leq \pricespot(\supply)} \pricespot^{s+r}(\supply) \typePDF(v)\\
&= \pricespot^{s+r}(\supply) (\spotCDF(\pricespot^{s+r}(\supply)) -\spotCDF(\pricespot^{s}(\supply)) )  \nonumber\\
&\leq \pricespot^{s+r}(\supply) \totalreserved(\pricespot^{s+r}(\supply))\\
&= \welfare_{\supply}^{-}(\mechspotandresbench). \label{eq:valunder}
\end{align}

Combining equations \eqref{eq:valunder} and \eqref{eq:valabove} gives:
\begin{align*}
\welfare_{\supply}(\mechspot) &\leq \welfare_{\supply}^{+}(\mechspotandresbench)  + \welfare_{\supply}^{-}(\mechspotandresbench)\\
	&=\welfare_{\supply}(\mechspotandresbench),
\end{align*}
our desired result.
\end{proof}

We can now combine the bounds from Claim \ref{claim:1} and Claim \ref{claim:2} to show that the welfare of the spot and reservation mechanism $\mechspotandres$ is greater than the welfare of the spot-only mechanism, $\mechspot$, completing the proof of Theorem \ref{thm:welfarebound}.

\begin{proof}[of Theorem~\ref{thm:welfarebound}]
Taking expectation over the supply and using the claims above, we have%
\begin{align*}
\welfare(\mechspot) \leq \welfare(\mechspotandresbench) \leq \welfare(\mechspotandres)
\end{align*}
as desired.
\end{proof}

\end{document}